\documentclass[twoside,12pt]{nictatechreport}
\usepackage{amsmath}
\usepackage{amssymb}
\usepackage{amsthm}
\usepackage{xparse}
\usepackage{xspace}
\usepackage[T1]{fontenc} 
\usepackage{xcolor}
\usepackage{subcaption}

\NewDocumentCommand{\ipic}{mO{scale=0.4}}{
	\includegraphics[#2]{#1.pdf}
}

\NewDocumentCommand{\pic}{O{}mO{scale=0.4}}{
	\begin{figure}[h]
		\center
		\ipic{#2}[#3]
		\caption{\label{#2}	#1}
	\end{figure}
}

\NewDocumentCommand{\pics}{omo}{
	\begin{figure}[h]
		\center
		#2
		\IfValueTF{#1}{\caption{#1 
			\IfValueTF{#3}{\label{#3}}{}
		}}{}		
	\end{figure}
}

\NewDocumentCommand{\subpic}{O{scale=0.4}mm}{
	\subcaptionbox{#2 
		\label{#3}		
	}{\ipic{#3}[#1]}
}

\NewDocumentCommand{\sym}{smO{}mo}{
	\expandafter\NewDocumentCommand\csname #2\endcsname{#3}{\ensuremath{{#4}}\xspace}
}

\NewDocumentCommand{\name}{mmm}{
	\expandafter\NewDocumentCommand\csname #1n\endcsname{}{#2\xspace}
    \expandafter\NewDocumentCommand\csname #1e\endcsname{}{\emph{#2}\xspace}
    \expandafter\NewDocumentCommand\csname #1u\endcsname{}{\capitalisewords{#2}\xspace}
	\expandafter\NewDocumentCommand\csname #1a\endcsname{}{#3\xspace}	
	\expandafter\NewDocumentCommand\csname #1ns\endcsname{}{#2s\xspace}
	\expandafter\NewDocumentCommand\csname #1as\endcsname{}{#3s\xspace}	
}

\newtheorem{lemma}{Lemma}
\newtheorem{corollary}{Corollary}
\newtheorem{theorem}{Theorem}
\newtheorem{definition}{Definition}

\name{optsol}{optimal solution}{OS}
\name{dc}{Linear DC network}{LDC network}
\name{msf}{maximum switching flow}{MSF}
\name{ots}{optimal transmission switching flow}{OTS}
\name{feas}{feasibility problem}{FEAS}
\name{feassol}{feasible solution}{swFEAS}
\name{sch}{choice network}{SCN}

\sym{sch}[O{x}O{}]{SCN_{#1}^{#2}}
\sym{msf}[O{\net}O{}]{\mathcal{\msfa}_{#2}(#1)}[the \msfn]
\sym{ots}[O{\net}O{}]{\mathcal{\otsa}_{#2}(#1)}[the \otsn]
\sym{func}[mmm]{#1: #2 \rightarrow #3}
\sym{reals}{\mathbb{R}}[real numbers]
\sym{preals}{\reals^{+}}[real numbers, positive]
\sym{naturalnumbers}{\mathbb{N}}[natural numbers]
\sym{buses}[O{}O{}]{N_{#1}^{#2}}[buses]
\sym{net}[O{}O{}]{\mathcal{N}_{#1}^{#2}}[network]
\sym{ba}[O{}]{i_{#1}}[bus]
\sym{bb}[O{}]{j_{#1}}[bus]
\sym{ca}[O{}]{s_{#1}}[line capacity]
\sym{su}[O{}]{b_{#1}}[line susceptance]
\sym{alllines}[O{\buses}]{\mathcal{\lines}({#1})}
\sym{sedge}[O{\ba\bb}]{[#1]}
\sym{edge}[O{\ba\bb}O{\ca}O{\su}]{\sedge[#1]_{#2}^{#3}}[line between \ba and \bb with capacity \ca and susceptance \su]
\sym{pflowu}{P}
\sym{gensy}{G}
\sym{loadsy}{L}
\sym{pgen}[O{}]{\gensy_{#1}}[generation]
\sym{pgmax}[O{}]{\overline{\gensy}_{#1}}[generation upper bound]
\sym{pload}[O{}]{\loadsy_{#1}}[load]
\sym{plmin}[O{}]{\underline{\loadsy}_{#1}}[load lower bound]
\sym{plmax}[O{}]{\overline{\loadsy}_{#1}}[load upper bound]
\sym{pas}{\Theta}
\sym{pa}[O{}]{\pas_{#1}}[phase angles]
\sym{lines}[O{}O{}]{E_{#1}^{#2}}[lines]
\sym{dlines}[O{\lines}]{\hat{#1}}[lines, directed version]
\sym{linesw}[O{}]{\lines'}[lines, switched]
\sym{dlinesw}[]{\dlines[\linesw]}
\sym{costs}[O{}]{c_{#1}}[generation costs]
\sym{pflow}[O{}]{p_{#1}}[line flow]
\sym{feassols}[O{\net}]{\mathcal{T}_{#1}}[set of \feasn]
\sym{feas}[O{\net}O{}]{{\feasa}_{#2}(#1)}
\sym{feassol}{(\linesw, \pa, \pgen, \pload)}

\begin{document}

\title{The Complexity of DC-Switching Problems}
\author{
    Karsten Lehmann 
    \textsuperscript{1,2}
    Alban Grastien
    \textsuperscript{1,2}
    Pascal Van Hentenryck
    \textsuperscript{2,1}
}

\affiliation{
    \textsuperscript1
    Optimisation Research Group, National ICT Australia, Canberra Laboratory\\
    \textsuperscript2
    Artificial Intelligence Group, Australian National University\\
    \email{first.last@nicta.com.au}
}

\headertitle{Lehmann, K. DC-Switching Technical Report}
\submitdate{}
\publishdate{}
\reportnumber{ZZZZ}
\copyrightyear{2014}

\begin{abstract}
    This report provides a comprehensive complexity study of
    line switching in the Linear DC model for the
    \feasn and the optimization problems of  
    maximizing the load that can served (maximum switching flow, MSF) 
    and
    minimizing generation cost (optimal transmission switching, OTS).
    Our results show that these problems are NP-complete and that there is no fully polynomial-time approximation scheme for planar networks with a maximum-node degree of $3$. 
Additionally, we demonstrate that the problems are still NP-hard if we restrict the network structure to cacti with a maximum degree of $3$.
We also show that the optimization problems can not be approximated within any constant factor.
\end{abstract}

\frontmatter
\tableofcontents
\mainmatter


\section{Introduction}
In this report we use the commonly accepted Linear DC power network model and consider reconfiguration via \emph{line switching} for the optimization problems:
maximizing the load served (\msfn, \msfa)
and 
minimizing the total generation costs (\otsn, \otsa).
Contained within both problems is the \feasn (\feasa), deciding whether a given demand can be satisfied.
We refer to these problems as the \emph{switching problems}.
Line switching means physically disconnecting  two buses that were previously connected  (or vice versa, connecting two previously disconnected buses).

The report provides the proof that the \msfa and the \otsa problem can not be approximated within any constant factor.
Real world power networks are not arbitrary graphs: for instance, their maximum bus degree is limited and they are (almost) planar networks, hence the class of real world power networks could still be easy to solve.
Consequently, we study the complexity of the switching problems for the class of planar graphs with a maximum degree of 3.
We establish the result that the \feasn is strongly NP-hard for planar graphs with a maximum degree of $3$.

For tree networks, the switching problems are easy because the complexity is driven by cycles.
In the absence of cycles, there are no cyclic dependencies on the phase angles. 
Hence they can be chosen in such a way as to match any optimal solution of the traditional max flow. 
Cacti are an obvious relaxation of trees.
They allow for cycles but every line can only put constraints on at most one cycle.
The switching problem for cacti is NP-hard.

N-level Tree networks are another possible relaxation of trees.
An n-level Tree network is an \dca based on a tree where there is one generator at the root and loads at the leaves.
Lines that are not part of the tree can only be added between buses on the same tree level where the level is less or equal to $n$ and only such that the resulting graph is planar.
This network structure is motivated by the disaster management application.
After the destruction of many power lines, it is easier to first repair lines such that we obtain a tree structure.
Then we can start restoring additional lines.
We show that the \msfa and the \otsa problem for 2-level tree networks is NP-hard.

Table~\ref{tbl:overview} presents an overview of the results contained in this document.
The number $3$ after cacti and planar stands for ``with maximum degree of 3''.
Note that the \feasn is a sub-problem of \msfa and \otsa.
It should be noted that the complexity results we present use unrealistic network parameters for simplicity.
The values can all be scaled to be realistic without influencing the results.

\begin{table}[h]
    \center
    \begin{tabular}{|c|c|c|c|c|c|}
        \hline
       problem & graph structure & loads & generators & complexity & reference \\
    \hline \hline
    \feasa & cacti3 & * & * & NP-hard & Theorem~\ref{theo:cacti} \\
     \msfa & cacti3 & * & * & NP-hard& Lemma~\ref{lem:feas}\\
     \otsa & cacti3 & * & * & NP-hard& Lemma~\ref{lem:mots}\\
     \msfa & 2-level tree& * & 1 & NP-hard & Theorem~\ref{theo:tree}\\
     \otsa & 2-level tree& * & 1 & NP-hard & Lemma~\ref{lem:mots}\\
     \feasa & planar3 & 1 & 1 & strongly NP-hard & Theorem~\ref{theo:planar}\\
     \msfa & planar3 & 1 & 1 & strongly NP-hard & Lemma~\ref{lem:feas}\\
     \otsa & planar3 & 1 & 1 & strongly NP-hard & Lemma~\ref{lem:mots} \\
     \msfa & arbitrary & 2 & 2 & non-APX & Theorem~\ref{theo:approx}\\
     \otsa & arbitrary & 1 & * & non-APX & Theorem~\ref{theo:approxots}\\
    \hline
    \end{tabular}
    \caption{Result Overview}
    \label{tbl:overview}
\end{table}

\section{Definitions}
This section presents the network model, the power flow equations and defines the switching problems this report is concerned with.

\subsection{\dcn}
In this report we use the \emph{Linear DC (LDC) model} of electrical power networks \cite{Schweppe_1970_Powersystemstatic}.
The LDC model is a linearization of the nonlinear steady-state electrical power flow equations (Alternating Current Model) and is widely used in practice.  
It assumes that all voltage magnitudes are one in the per-unit system and ignores reactive power and resistance which are small relative to real power and reactance during normal operations. 
What is left is the susceptance (a physical parameter of networks)\footnote{Susceptance is a negative value but when used to calculate the flow it is multiplied by $-1$. For readability, we omit the $-1$ and make the susceptance a positive value.}
, the capacity and the phase angles of the voltages.
The flow that is transmitted by a line is the product of the phase angle difference between its two ends and the susceptance.  

\begin{definition}
    Let \buses be a set.
    Given the set $A = \{(\{\ba,\bb\}, \ca, \su) \mid \ba, \bb \in \buses; \ca \in \preals; \su \in \preals\}$
    we define 
    $\alllines := \{ T \subseteq A \mid 
        \forall (\{\ba,\bb\},\ca[1],\su_1), (\{\ba,\bb\},\ca_2,,\su_2) \in T: \ca[1]=\ca_2, \su_1=\su_2
    \}$.
\end{definition}

\begin{definition}
	A \dce (\dca) is a tuple $\net = (\buses, \lines, \plmin, \plmax, \pgmax, \costs)$ where 
    \buses is the set of buses;
    $\lines \in \alllines[\buses]$ is the set of lines; 
\func{\plmin}{\buses}{\preals}, \func{\plmax}{\buses}{\preals} are the minimum and maximum demand;
\func{\pgmax}{\buses}{\preals} is the maximum generation;
and
\func{\costs}{\buses}{\preals} are the generation costs.
\end{definition}

\subsection{Power Flow Equations}
We now introduce the notations and equations pertaining to \dca power flows.
The following definitions assume a fixed \dca $\net = (\buses, \lines, \plmin, \plmax, \pgmax, \costs)$.
The value of the function \func{\pa}{\buses}{\reals}, $\pa[\ba]$ is the \emph{phase angle} at bus \ba. 
Moreover the \emph{generation} and \emph{load} at a bus are given by functions \func{\pgen}{\buses}{\preals} and \func{\pload}{\buses}{\preals}. 
The directed version of the set of lines is denoted by
$\dlines := \{ (\ba,\bb,\ca,\su), (\bb,\ba,\ca,\su) \mid (\{\ba,\bb\},\ca,\su) \in \lines\}$.
We write \edge for a line from \ba to/between \bb with susceptance \su and capacity \ca.

The \emph{flow} on a line is given by \func{\pflow}{\dlines}{\reals}.

The \dca model imposes two sets of constraints: Kirchhoff's conservation law and the \dca power law.
Kirchhoff's conservation law states that the power that enters a bus equals the power that leaves this bus.  
The \dca power law binds together the power flow, the phase angle and the susceptance of a line.  
A \feassoln is a tuple \feassol that satisfies both laws and the generation and load bounds
where \linesw is the set of switched off lines.

\begin{definition}
    Given the tuple \feassol we define a flow \pflow via
    $\forall \edge \in \dlinesw: 
    \pflow[\sedge] := \su(\pa[\bb] - \pa[\ba])
    $.
    We call \feassol \emph{\feassoln} if
    \begin{itemize}
        \item $\linesw \subseteq \lines$; 
        \item $\forall \edge \in \dlines \setminus \dlinesw: |\pflow[\sedge]| \leq \ca $;
        \item $\forall \ba \in \buses:
    \plmin[\ba] \leq \pload[\ba] \leq \plmax[\ba]$;
        \item $\forall \ba \in \buses:
    0 \leq \pgen[\ba] \leq \pgmax[\ba]$;
        \item $\forall \ba \in \buses:
    \sum_{\sedge \in \dlines \setminus \dlinesw} \pflow[\sedge] = \pgen[\ba] - \pload[\ba]$.
    \end{itemize}
    The set of all \feassolns is denoted with \feassols.
\end{definition}

\subsection{The Switching Problems}
This report presents complexity results for three types of problems: minimizing generation dispatch; maximizing the served demand; and the question of whether or not a given demand can be satisfied.
In all three cases we consider that lines can be switched off.

\begin{definition}
    Let \net be an \dca. The \emph{\msfn} (\msfa) is defined as
    $$\msf := \max_{\feassol \in \feassols} \sum_{\ba \in \buses} \pload[\ba].$$
    The \emph{\otsn} (\otsa) is defined by
    $$\ots := \min_{\feassol \in \feassols} \sum_{\ba \in \buses} \costs[\ba]\pgen[\ba].$$
    The \feasn \feas is to decide whether or not $\feassols \neq \emptyset$.
\end{definition}
To establish our complexity results, we define a decision version of our optimization problems.
\begin{definition}
    Given an $x \in \preals$. The \emph{\msfa (resp. \otsa) problem} is the problem of deciding if $\msf \geq x$ (resp. $\ots \leq x$).
\end{definition}

\subsection{Examples and Graphical Representation:}
\pics[Examples for MPF and \msfa.]{
	\subpic{An \dca.}{triangle}
    \subpic{The max flow.}{triangle_mf}
	\subpic{The MPF.}{triangle_mpf}
	\subpic{The \msfa.}{triangle_msf}	
}[pic:examples]
Figure~\ref{pic:examples} introduces our graphical representations for \dcas along with examples for the maximum potential flow (MPF) and \msfa. 
The MPF is the \msfa without the switching.
We omit the susceptance and/or capacity of a line when its value is $1$.
Figure~\ref{triangle} shows an \dca where $g$ is a generator (box), $l$ is a load (house) and $b$ is a normal bus (sphere). 
It is easy to see that the traditional max flow for this network is $34$ whereas in the LDC model, we only can supply $16$ as shown in Figure~\ref{triangle_mpf} because the congestion of the line \edge[bl][1][4] constrains the phase angle (written as $A=$ in the buses) between $g$ and $l$. 
However, by switching the line \edge[gb][1][5], we can improve the maximum generation to $30$ as shown in Figure~\ref{triangle_msf}.

\subsection{Notional Conventions}
For all the following proofs we set the convention that any value for any of the functions \plmin, \plmax, \pgmax, \costs, \pa, \pgen and \pload that is not specifically defined is $0$.
When defining networks we do not specify the costs \costs if they are not of interest.
Whenever we write that a bound is infinite, we mean that having a bound does not matter to make the proof work.
The reader can assume that the actual value is sufficiently large to not influence the result of the proof.

\section{Relationship between the Switching Problems}
In this section we present the results that the \feasn is a sub-problem for the \msfa problem; and the \msfa problem is a sub-problem of the \otsa problem.
Hence every \feasn result does also apply for \otsa and \msfa.
\begin{lemma}
    \label{lem:feas}
    Let $\net = (\buses, \lines, \plmin, \plmax, \pgmax)$ be a \dca.
    Then we have  
    $\msf \geq \sum_{\ba \in \buses} \plmin[\ba] \iff \feassols \neq \emptyset$.
\end{lemma}
\begin{proof}
    If there is a \feassoln then we can satisfy the minimum demand on every bus and hence 
    $\msf \geq \sum_{\ba \in \buses} \plmin[\ba]$. 
    On the other hand, the \msfa can only be
    $\msf \geq \sum_{\ba \in \buses} \plmin[\ba]$ if there exists a \feassoln because otherwise there is no flow possible at all.
\end{proof}

We can encode every \msfa problem into an \otsa problem such that they have the same network structure.
Let us assume that no load has a generator.
This is done as follows.
We set the costs of every generator to $0$ and fix the demand.
Then we place a generator with cost 1 at every load.
These generators get a maximum generation equal to the difference between the maximum and the minimum demand.
That way we pay a cost of $1$ for every $1$ demand we can not satisfy via the original generators.
This encoding implies that any hardness result for \msfa is also true for the \otsa.
However it does not apply to approximation results.

\begin{lemma}
    \label{lem:mots}
    Let $\net = (\buses, \lines, \plmin, \plmax, \pgmax)$  be a \dca network with
    $\pgmax[\ba] > 0 \implies \plmax[\ba] = 0$ and $\plmax[\ba] > 0 \implies \pgmax[\ba] = 0$.
    We define 
    $\net' = (\buses, \lines, \plmin', \plmax, \pgmax', \costs)$  with
    $\plmin[\ba]' := \plmax[\ba]$;
    $\plmax[\ba] > 0 \implies \pgmax[\ba]' := \plmax[\ba]-\plmin[\ba], \costs[\ba] := 1$.
    We have
    $\sum_{\ba \in \buses} \plmax[\ba] - \msf[\net] = \ots[\net'].$
\end{lemma}

\section{Cacti Networks}
In this section we present the result that the \feasn for cacti networks with a maximum degree of 3 is NP-hard.
To show this result, we first provide a type of network called \sche that allows us to encode decisions.

\begin{definition}[\schn, \scha]
Let $x \in \preals$. 
The \emph{\schn for $x$ with connector $v$} 
is the \dca
$\sch[x,v] := (\buses, \lines, \plmin, \plmax, \pgmax, \costs)$
define with
$\buses := \{g, l , v\}$;
$\lines := \{\edge[gl][2x][1], \edge[gv][x][1], \edge[vl][x][1] \}$;
$\plmin[l] := \plmax[l] := 3x$
$\pgmax[g] := 3x$.
Let \sch[x,v][+] be the version of \sch[x,v] where $v$ is an unbounded generator, $\plmax[v] = \infty$.
and 
\sch[x,v][-] be the version of \sch[x,v] where $v$ is an unbounded load, $\plmin[v] = 0$ and $\plmax[v] = \infty$.
\end{definition}

\pics[The \schn.]{
    \subpic{{\sch[x,v][-]}}{gsch}
    \subpic{$\pgen[v] = 0$}{gsch_msf_1}
    \subpic{$\pgen[v] = x$}{gsch_msf_2}
}[fg:gsch]
Figure~\ref{fg:gsch} shows a \schn and two solutions.
Figure~\ref{lem:sch} states that the \schn \sch[x,v][+] has two (significantly different) \feassolns; the solutions presented in Figure~\ref{fg:gsch} are optimal and the only ones possible (w.r.t. phase angle offsets); and that whenever the bus $v$ acts as a load
then there is no \feassoln.

\begin{lemma}
\label{lem:sch}
Let $x \in \preals$ and \sch[x,v] be the \scha.
We have 
\begin{align}
    \label{lem:sch:load}\{ \feassol \in \feassols[\sch[x,v][-]] \mid \pload[v] > 0\} &= \emptyset;\\
    \label{lem:sch:gen}\{ \pgen[v] \mid \feassol \in \feassols[\sch[x,v][+]]\} &= \{0, x\}.
\end{align}
\end{lemma}
\begin{proof}
    Let \feassol be a \feassoln of \sch[x,v][+] or \sch[x,v][-]. 
    W.l.o.g. we assume that $\pa[l] = 2x$.
    The load has two lines with a total capacity of $3x$ and a demand of $3x$.
    Hence we can derive that all lines from $l$ must be congested and can therefore not be switched in any \feassoln.
    This implies flows $\pflow[gl] = 2x$ and $\pflow[vl] = x$
    and phase angles $\pa[g] = 0$ and $\pa[v] = x$.

    Let us assume that $\linesw = \emptyset$.
    The phase angles imply a flow of $x$ from $g$ to $v$.
    Kirchhoff's conservation law at $v$ is
    $\pflow[vg] + \pflow[vl] = -x + x = 0 = \pgen[v] - \pload[v]$.
    Hence if $v$ is a generator or a load it can neither generate nor consume power.

    Let us assume that $\linesw = \{\sedge[gl]\}$.
    Kirchhoff's conservation law at $v$ is
    $\pflow[vl] = x = \pgen[v] - \pload[v]$
    In order to fulfil this equation $v$ must be a generator with a generation of $x$.
\end{proof}

Using \schns we can show that cacti networks are NP-hard.

\begin{theorem}
    \label{theo:cacti}
    The \feasn for cacti networks with a maximum degree of 3 is NP-hard.
\end{theorem}
\begin{proof}
The proof is done by reduction from the \emph{Subset sum problem}. 
Given a finite set $M \subset \naturalnumbers$ and a $w \in \naturalnumbers$:
the {Subset sum problem} is to decide whether there is a $V \subseteq M$ such that $\sum_{x \in V} x = w$. 
If such a subset exists then we call the problem $(M,w)$ \emph{solvable}.

Let $(M,w)$ be an instance of the Subset sum problem with $M = \{ x_1, \ldots x_n \}$.
To encode the problem we use the network 
$\net[*][M,w] := (\buses, \lines, \plmin, \plmax, \pgmax)$ 
where 
$\buses := \{g, l, x_0\} \cup \bigcup_{1 \leq i \leq n} \{v_i, x_i\}$ 
$\lines := \{ \edge[gl][2+w][1], \edge[gv_0][w+1][1], \edge[v_0l][1][1]\} \cup \bigcup_{0 \leq i < n} \{ \edge[x_ix_{i+1}][w][1], \edge[x_{i+1}v_{i+1}][\infty][1]\}$;
$\pgmax[g] = 2+w$;
$\plmin[l] := \plmax[l] := 3 + w$;
We define $\net[M,w] := \net[*][M,w] + \sum_{1 \leq i \leq n} \sch[x_i,v_i]$ and we have: 
$$\feassols[\net[M,w]] \neq \emptyset \iff (M,w) \text{ is solvable}.$$
Figure~\ref{pic:cactus} shows an example encoding for $(M,w) = (\{1,2,3\}, 5)$ where 
the \sch[x_i,v_i] presented as black boxes.
\definecolor{two}{HTML}{40E0D0}
\definecolor{fuchs}{HTML}{FF00FF}
\definecolor{olive}{HTML}{808000}
\definecolor{orang}{HTML}{FFA500}
\definecolor{brow}{HTML}{A52A2A}	
\pics[Example for $(M,w) = (\{\textcolor{olive}{1}, \textcolor{two}{2}, \textcolor{fuchs}{3} \}, \textcolor{orang}{5})$]{
    \subpic{The network {\net[M,w]}}{cactus2}
    \subpic{A solution for {\net[M,w]}}{cactus2_mf}
}[pic:cactus]

It is easy to see that this reduction is polynomial and because the networks \sch[x,v] are cacti, that the network \net[M,w] is a cactus.\\

\noindent
Case 1: $\feassols[\net[M,w]] \neq \emptyset \implies (M,w) \text{ is solvable}$.
Let \feassol be a \feassoln of \net[M,w].
The sum of capacities of all lines connected to $l$ is $3 + w$.
Hence, to satisfy the demand of $3 + w$ all lines have to be congested.
This implies phase angle differences $\pa[g] - \pa[l] = w+2$ and $\pa[x_0] - \pa[l] = 1$.

Assume that the line \sedge[gx_0] is switched, so $\sedge[gx_0] \in \linesw$.
Because there is a flow of $1$ on the line \sedge[x_0l]
there has to be a flow of $1$ coming from one of the \schns.
If we regard the \schns as black boxes then
Lemma~\ref{lem:sch:load} from Lemma~\ref{lem:sch} implies that none of the \schns can act as a generator.
Hence we have a contradiction.
That implies that the line \sedge[gx_0] is not switched.
For its flow we have 
$\pflow[\sedge[x_0g]] = \pa[g] - \pa[x_0] = \pa[g] - \pa[l] - (\pa[x_0] - \pa[l]) = w+2-1 = w+1$.
Because there is a flow of $1$ going from $x_0$ to $l$ there is a flow of $w$ going towards the \schns.
Lemma~\ref{lem:sch:gen} implies that in order to satisfy the demand of the \schns the buses $v_i$ ($i>1$) 
have to consume either $0$ or $x_i$.
We call a \schn which consumes $x$ \emph{active}.
Let $V := \{ x \in M \mid \sch[x_i,v_i] \text{ is active}\}$.
Because \feassol is a \feassoln we know that all Kirchhoff's conservation laws are satisfied and hence the \schns have to consume the flow of $w$ coming from $x_0$.
Therefore we have $\sum_{x \in V} = w$.

\noindent
Case 2: $\feassols[\net[M,w]] \neq \emptyset \Longleftarrow (M,w) \text{ is solvable}$.
Let $V \subseteq M$ with $\sum_{x \in V} = w$.
Our observations above imply that by activating all \schns \sch[x_i, v_i] with $x_i \in V$ we get a \feassoln.
\end{proof}

Cacti networks are simpler than real power networks.
From the results of this section, we conclude that the switching problem are also hard for real world power networks.

\section{N-level Tree Networks}
N-level Tree networks are a possible relaxation of trees.
In this section we show that the \msfa problem for 2-level tree networks are NP-hard.
Lemma~\ref{lem:mots} then implies that this result is also true for the \otsa problem.
An n-level Tree network is an \dca based on a tree where there is one generator at the root and loads at the leaves.
Lines that are not part of the tree can only be added between buses on the same tree level where the level is less or equal to $n$ and only such that the resulting graph is planar.
This network structure is motivated by the disaster management application.
After the destruction of many power lines, it is easier to first repair lines such that we obtain a tree structure.
Then we can start restoring additional lines.
\begin{definition}[n-level Tree network]
	Let $n \in \naturalnumbers$. 
    An \emph{n-level Tree network} is an \dca iff there exists a sub-network $T$ that is a Tree such that:
    all leaves of $T$ are loads;
    there is only one generator at the root node of $T$ and
    there is a total order on the children of every node (which implies a total order on all nodes in one level) such that
    every node of the same tree level can only be connected to its neighbours in the total order on all nodes of the same level.
\end{definition}

\begin{theorem}
    \label{theo:tree}
The \msfa problem for 2-level Tree networks is NP-hard.
\end{theorem}
\begin{proof}
    We prove this reduction from a version of the \emph{subset sum} problem.
    Given an instance $(M, w)$, let $m := 1 + \sum_{x \in M} x$ and $M = \{a_2, \ldots, a_n\}$.
    We use $a_i$ to represent a value from $M$ as well as a symbol corresponding to that value.
    The network \net[][M,w] is defined by 
    $\net[M,w] := (\buses, \lines, \plmin, \plmax, \pgmax)$ with 
    $\buses := \{g\} \cup \{l_i \mid 1 \leq i \leq n+1 \} \cup \{g_1, g_{n+1}, t, a_1, a_{n+1} \} \cup M$; 
    $\lines := 
    \bigcup_{1 \leq i \leq n} \allowbreak\{ \edge[pa_i][a_i][\frac{a_i}{i}],\allowbreak \edge[a_il_i][a_i][1], \allowbreak\edge[a_{i-1}a_i][m][m] \} 
    \cup
    \{ \edge[gg_1][m+1][2m+2], \edge[g_1a_1][m+1][2m+2],\allowbreak \edge[a_1l_1][1][1], \edge[gt][w][w], \edge[gg_{n+1}][1][\frac{2}{n+1}],\allowbreak \edge[g_{n+1}a_{n+1}][1][\frac{2}{n+1}], \edge[a_{n+1}l_{n+1}][m + 1][1], \edge[a_na_{n+1}][m][m] \}$;
    $\pgmax[y] = \infty$;
    $\forall 1 \leq i \leq n+1: 
    \plmax[l_i] := a_i$.
    We have 
    $$\msf[\net[M,w]] = m + 2 + w \iff (M,w) \text{ is solvable}.$$
    
    An example encoding for $(M,w) = (\{2,1,3\}, 5)$ can be found in Figure~\ref{theo:tree}.
    \pics[Example for $(M,w) = (\{2,1,3\}, 5)$]{
        \subpic{The network \net[M,w]}{treelevel}
        \subpic{A solution for \net[M,w]}{treelevel_msf}
    }[pic:tree]
	It is easy to see that \net[M,w] is a 2-level Tree network. 

	\textbf{Case 1:}$\msf[\net[M,w]] = m + 2 + w \implies (M,w) \text{ is solvable}$.
	Let \feassol be an optimal solution of the \msfa. 
    W.l.o.g. let $\pa[g] = 0$. 
    Since the max flow is $m + 2 + w$ we know that the lines \sedge[gg_1], \sedge[g_1a_1], \sedge[gg_{n+1}], \sedge[g_{n+1}a_{n+1}] and \sedge[gt] are congested and therefore $\pa[a_1] = \pa[t] = 1$ and $\pa[a_{n+1}] = 2 \frac{n+1}{2} = n+1$. 
    This implies that we have at least $m + 2$ incoming power at the bus $a_1$. 
    Since the other two lines have in sum a capacity of $m + 2$ we know that they are congested. 
    Therefore, we obtain $\pa[a_2] = 2$. 
    For the bus $a_n$ we know that the phase angle can not be bigger then $n$ because that would overload the line \edge[ta_i][a_i][\frac{a_i}{i}]. 
    However, if the phase angle is smaller then $n$, then the line \edge[a_na_{n+1}][m][m] is overloaded. 
    Therefore, $\pa[a_n] = n$. 
    We can apply similar arguments to $a_{n-1}$ to have a phase angle of $n-1$. 
    Overall, we derive that $\forall 1 \leq i \leq n: \pa[a_i] = i$. 
    Hence, the lines \edge[a_{i-1}a_i][m][m] for $1 \leq i \leq n+1$ must be congested. 
    They also cannot be in \linesw since their flow $m$ is greater then the sum of elements of $M$ an therefore the sum of power we can send to the loads $l_i$ and can get from the buses $a_i$  with $1 \leq i \leq n$.
		
	We define $V := \{ a_i \in M \mid \sedge[ta_i] \notin \linesw \}.$ We know that the incoming power at bus $t$ is $w$, $t$ respects Kirchhoff's conservation law and that all lines $\sedge[ta_i] \notin \linesw$ are congested. Therefore $\sum_{x \in M} x = w$.\\
	\newline
	\textbf{Case 2:}$(M,w) \text{ is solvable} \implies \msf[\net[M,w]] = m + 2 + w$.
	Let $V$ be a solution of $(M,w)$. We define phase angle \pas with $\pa[g] := 0, \pa[g_1] := \frac{1}{2}, \pa[g_{n+1}] := \frac{n+1}{2}, \pa[a_i] := i, \pa[l_1] := 2, \pa[l_{n+1}] ;= n+2+m$ and 	$\forall 1 \leq i \leq n:$
	$$\pa[l_i] := 
		\begin{cases}
			i + a_i		&\text{if } a_i \in V\\
			i			&\text{otherwise}.
		\end{cases}
	$$
	We also define $\linesw := \{ \sedge[ta_i] \in \lines \mid a_i \notin V \}$. Since the sum of all elements of $V$ is $w$, we know that with this definition $t$ respects Kirchhoff's conservation law and it is easy to see that all other buses do the same. This definition also implies a flow of $m + 2 + w$ which is the \msfa because all lines of the generator are congested.
\end{proof}

\section{Approximation Complexity}
This section presents the proof that the \msfa and the \otsa problem can not be approximated within any constant factor (they are non-APX).
First we present a reduction of the Longest Path problem to the \msfa problem such that every approximation algorithm for the \msfa would imply an approximation result for Longest Path.
This reduction implies that the switching problem can not be approximated in polynomial time within any constant factor.
The Longest Path problem is: given two buses $a$ and $b$ in a graph, find a path that starts in $a$, ends in $b$ and visits every bus at most once and is maximal in terms of the number of buses visited.
It is known that for all $\epsilon >0$, it is not possible to approximate Longest Path to within a factor of $2^{(\log n)^{1-\epsilon}}$ unless NP is contained within quasi-polynomial deterministic time \cite{karger1997approximating}.

\begin{theorem}
\label{theo:approx}
It is not possible to approximate the \msfa problem within a factor of $2^{(\log n)^{1-\epsilon}}$ unless NP is contained within quasi-polynomial deterministic time.
\end{theorem}
\begin{proof}
The proof is done by reduction from the $a-b$ Longest Path problem.
Given a graph $G = (\buses[h],\lines[h])$ with $\buses[h] = \{v_1, \ldots, v_n\}$, $a=v_1$ and $b=v_n$. 

Let $\net[G] = (\buses, \lines, \plmin, \plmax, \pgmax)$ with
$\buses := \buses[h] \cup \{g, l, g', l'\}$;
$\pgmax[g] := \pgmax[g'] := \infty$;
$\plmax[l] := \plmax[l'] := \infty$;
$\plmin[l] := 3$;
$\lines := \{ \edge[cd][1][1] \mid \sedge[cd] \in \lines[h]\} \cup \{\edge[ga][1][1], \edge[bl][1][1], \edge[gl][\frac{1}{n+1}][1], \edge[g'l'][n][n], \edge[g'l][1][1], \edge[l'l][1][1]\}$.
\pic[Example for {\net[G]} with $G = (\{a, b, c, d\}, \{\edge[a][c], \edge[a][d], \edge[b][c], \edge[b][c], \edge[c][d]\})$.]{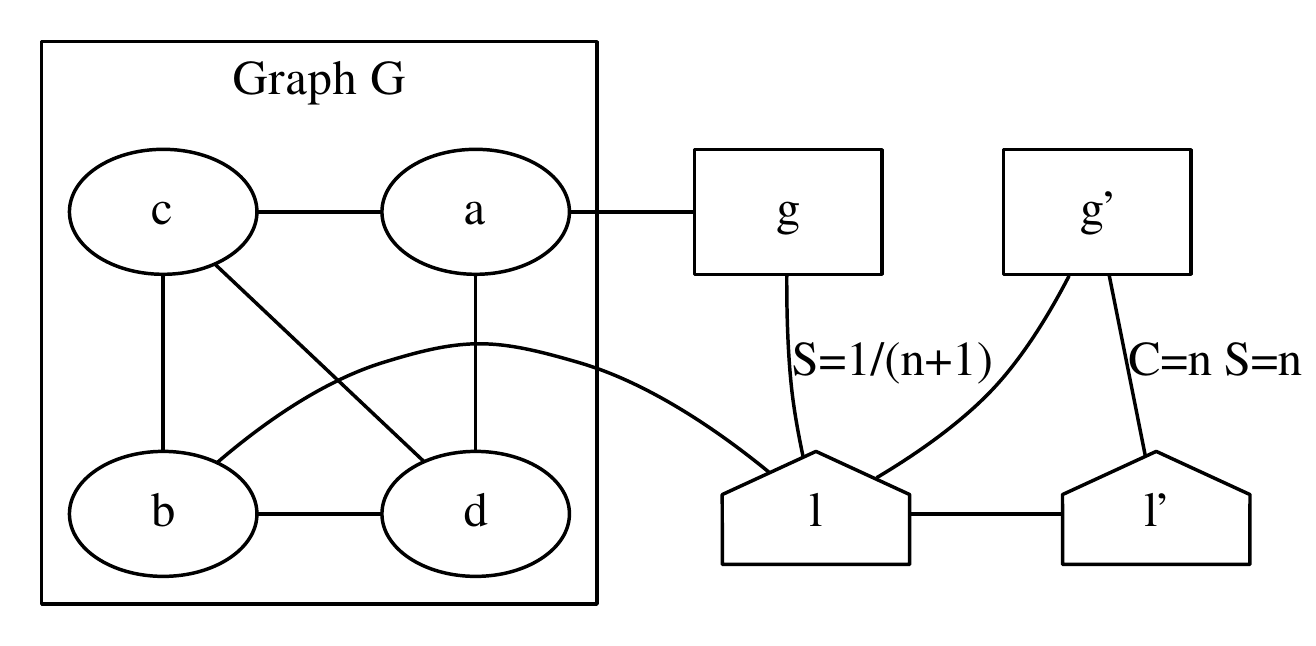}

Before we present the formal proof, we outline the idea.
It holds that for every $\epsilon$-approximation algorithm for the $MSF$ problem and every $1 > \theta > 0$ there is an $\epsilon \theta$-approximation algorithm for the Longest Path problem.
The connection between Longest Path and the \msfa of \net[G] is as follows.
The line \sedge[g'l'] is the line that delivers the majority of power in this network.
To maximize its flow, we have to maximize the phase angle difference between $g'$ and $l'$.
The phase angle difference is limited by the path from $g'$ through $l$ to $l'$.
The load $l$ has minimum demand of $3$. 
This demand has to be satisfied by power coming from $g$ and $g'$.
The phase angle difference is the bigger the less power we have to deliver from $g'$ to $l$.
To maximize the flow along the line \sedge[gl] we have to maximize the phase angle difference between $g$ and $l$.
Switching off lines in $G$ allows to increase the phase angle difference.
It is maximal when switching off lines in $G$ results in a sub-graph that is a Longest Path.
Any $\epsilon$-approximation for the \msfa would produces a switching in $G$. Using the phase angle difference between $g$ and $l$ we can derive a lower bound for the Longest Path.

The minimum demand at $l$ implies 
\begin{equation}
    \label{u}
    3 \leq \pflow[gl] + \pflow[bl] + \pflow[g'l] + \pflow[l'l].
\end{equation}
Hence $\pflow[gl] + \pflow[bl] \geq 1$.

Let us assume that $\pflow[gl] + \pflow[bl] = 1$. This implies that \sedge[g'l] and \sedge[l'l] have to be congested. That implies that the phase angles of $g'$ and $l'$ are the same and therefore there is no flow on \sedge[g'l']. Hence Kirchhoff's Junction law will be violated at $l'$. Therefore we can assume that $\pflow[gl] + \pflow[bl] > 1$. This implies that the line \sedge[gl] will never be switched and that there is always a path from $g$ to $l$ via the lines of the graph $G$.

In every feasible solution, the switching of lines $G$ will create a subgraph $W = (\buses[h], \lines[W])$ with $\lines[W] \subseteq \lines[h]$. 
Our observations above prove that there is at least one path $p$ from $a$ to $b$ in $W$. 
Let $t$ be the length of $p$.
In the following we show upper bounds on the generation of $g$ and $g'$ depending on $t$.

The DC power law implies that every line in that graph allows for a maximum phase angle difference of $1$. Therefore, the maximum phase angle difference between $g$ and $l$ with respect to $p$ is bounded by $t+2$, so $\pa[g]- \pa[l] \leq t+2$. This implies
\begin{equation}
    \label{c}
\pflow[gl] \leq \frac{t+2}{n+1}
\end{equation}
and since the line \sedge[bl] has a capacity of $1$ we have 
\begin{equation}
    \label{d}
\pgen[g] \leq 1 + \frac{t + 2}{n+1}.
\end{equation}
We can achieve equality, if and only if $W$ contains only one path from $a$ to $b$. 

Given the capacities of $\pflow[bl]$ and $\pflow[g'l]$ are $1$ and Equation~\ref{u} we have
\begin{equation}
    \label{b}
    -\pflow[l'l] \leq \pflow[gl] + \pflow[bl] + \pflow[g'l] - 3 \leq \pflow[gl] - 1.
\end{equation}

Kirchhoff's Conservation Law for the triangle $l, g', l'$ is $\pflow[g'l] = \frac{\pflow[g'l']}{n} + \pflow[l'l]$ which together with Equation~\ref{b} and Equation~\ref{c} implies
\begin{equation}
\label{a}
\pflow[g'l'] = n(\pflow[g'l] - \pflow[l'l]) \leq n\pflow[gl] \leq n\frac{t+2}{n+1}.
\end{equation}

Finally, we have
\begin{equation}
    \label{e}
    \pgen[g'] = \pflow[g'l] + \pflow[g'l'] \leq 1 + n\frac{t+2}{n+1}
\end{equation}
where we can achieve equality if and only if $\pflow[bl] = 1$.

Combining Equation~\ref{d} and Equation~\ref{e} shows that 
\begin{equation}
    \label{t}
\pgen[g] + \pgen[g'] \leq 3 + t
\end{equation}
and we have equality if and only if $W$ consists of one path only.
The later implies that $\msf[\net[G]] = 3 + t_G$ where $t_G$ is the length of a longest path in $G$.

Let $1 > \epsilon > 0$ and lets assume that we have an $\epsilon$-approximation algorithm $A$ for \msf[\net[G]] that provides us with a feasible solution and its total generation $A(\net[G])$ such that 
$$\epsilon (3 + t_G) = \epsilon \msf[\net[G]] \leq A(\net[G]) \leq \msf[\net[G]].$$

Let $t$ be the length an arbitrary path from $a$ to $b$ through $G$ in the solution of $A$ that we can find in polynomial time (for example the shortest path).
Equation~\ref{t} implies that $\epsilon (3 + t_G) \leq A(\net[G]) \leq 3 + t$ and hence 
$$\epsilon t_G - 3 \leq \epsilon t_G - 3(1-\epsilon) \leq t.$$
This shows that for every $\epsilon$-approximation algorithm for the $MSF$ problem and every $1 > \theta > 0$ there is an $\epsilon \theta$-approximation algorithm for the longest path problem.
\end{proof}

It follows a reduction from the minimum 3-dimensional assignment to the \otsa problem.

\begin{theorem}
\label{theo:approxots}
The \otsa problem is non-APX.
\end{theorem}
\begin{proof}
    The proof uses a reduction or the \emph{Minimum 3-Dimensional Assignment} which is non-APX \cite{crama1996approximation}.
    Given three sets $X, Y, W$ and a cost function \func{d}{X\times Y\times W}{\naturalnumbers}.
    An assignment is a set $A \subseteq X \times Y \times W$ such that every element of $X \cup Y \cup W$ belongs to exactly one triple in $A$.
    The cost of an assignment $A$ is $\sum_{t \in A} d(t)$.
    The minimum 3-dimensional assignment is defined an assignment with minimum cost and we set $M3DA(X,Y,W,d)$ to be the costs of an optimal assignment.

    Let $(X,Y,W,d)$ be an instance of the minimum 3-dimensional assignment problem and let $l$ be a symbol that is not in $X \cup Y \cup W$.
    To simplify notations we define $T := X \times Y \times W$; $R := X \cup Y \cup W$ and for $t \in T$ and $r \in R$ the notation $r \in t$ represents $t = (x,y,w)$ and $r \in \{x,y,w\}$.
    We define the network $\net[X,Y,W,d] := (\buses, \lines, \plmin, \plmax, \pgmax, \costs)$ with
    $\buses := \{l\} \cup R \cup \{t, t_g, t_d  \mid t \in T\}$;
    $\forall t \in T: \pgmax[t] := 3, \pgmax[t_g] := 5, \costs[t] := \frac{d(t)}{3}, \costs[t_g] := 0$;
    $\plmin[l] := \plmax[l] := 3|T| + |R|$;
    and
    $\lines := 
    \bigcup_{r \in R} \{ \edge[rl][1][1] \}
    \cup
    \bigcup_{t \in T} ( \{ \edge[tt_g][3][1], \edge[tl][5][1], \edge[tt_d][3][1] \}
    \cup
    \bigcup_{r \in t} \{\edge[t_dr][1][1]\} )
    $.

    We are going to show that $\ots[\net[X,Y,W,d]] = M3DA(X,Y,W,d)$.
    This is done by showing that for every assignment there is a \feassoln and for every \feassoln there is an assignment such that the costs are the same.

    Let $A$ be an assignment.
    We define \feassol with 
    $\linesw := \{ \sedge[tt_d] \mid t \notin A\}$;
    $\pa[l] := 5$;
    $\pload[l] := 3|T| + |R|$;
    $\forall r \in R: \pa[r] := 4$;
    $\forall t \in T: \pa[t] := 0, \pa[t_g] := -5, \pa[t_d] := 3, \pgen[t_g] := 5$;
    $\forall t \in A: \pgen[t] := 3$.
    This definition satisfies all generation, load and flow bounds.
    Because $A$ is an assignment we know that every $r \in R$ is connected to exactly one $t \in T$.
    Hence Kirchhoff's conservation law at $r$ is satisfied.
    This shows that \feassol is a \feassoln.
    The cost of this solution is $\sum_{t \in T} (0\pgen[t_g] + \frac{d(t)}{3}\pgen[t]) = \sum_{t \in A} d(t)$ which is equal to the cost of the assignment.

    Let \feassol be a \feassoln.
    We define $A := \{ t \in T \mid \sedge[tt_d] \notin \linesw \text{ and } \exists r\in t: \sedge[rt_d] \notin \linesw \}$
    and show that $A$ is an assignment.
    To satisfy the demand, all lines of $l$ have to be congested and can not be switched off.
    Hence every $r \in R$ has an implicit demand of $1$ and must therefore be connected to some $t_d$ which has a connection to $t$.
    Assume there is an $r \in R$ with $t, t' \in A$, $t \neq t'$, $r \in t$ and $r \in t'$.
    Because $t\in A$ we know that the lines \sedge[tt_d] and \sedge[rt_d] are not switched.
    The congestion of \sedge[tl] and \sedge[rl] implies $\pa[r] - \pa[l] = 1$ and $\pa[t] - \pa[l] = 5$.
    Given these phase angle differences, the capacities of \sedge[tt_d] and \sedge[rt_d] allow only the solution $\pa[r] - \pa[t_d] = 1$ for the phase angle solution between $r$ and $t_d$.
    This implies a flow of $1$ from $t_d$ to $r$.
    Similarly we can derive that there is a flow of $1$ from $t_d'$ to $r$.
    Because the line \sedge[rl] has a capacity of $1$, the extra flow of $1$ must go to 
    some other $\hat{t}$ via $\hat{t_d}$ with $r \in \hat{t}$.
    However, because all lines of $l$ are congested we know that $\pa[r] - \pa[\hat{t}] > 0$ and hence no flow can go in that direction.
    This shows that the assumption $r \in t$ and $r \in t'$ leads to a contradiction and hence $A$ is an assignment.
    Every bus $t \notin A$ can satisfy its implicit demand from the line \sedge[tl] via the generator in $t_g$ for no costs.
    Hence $\ots[\net[X,Y,W,d]] = \sum_{t \in A} d(t)$.
\end{proof}

\section{Planar Networks}
In this section we show the complexity result that the \feasn is strongly NP-hard for planar networks with a maximum degree of 3.
For a given graph $G$, let \net[G] be as defined in the proof of Theorem~\ref{theo:approx}. 
We set \net[G][h] to be the same network where $g'$, $l'$ and all their lines are removed.
\net[G][h] provides an encoding of the Hamilton Path problem via the switching problem (without minimum demand).
One can use a very similar proof as in Theorem~\ref{theo:approx} to show that $\msf[\net[G][h]] = 2$ iff $G\text{ has a Hamiltonian Path}$.
The Hamiltonian circuit problem is strongly NP-complete even for planar and cubic graphs \cite{garey1976planar} and the generalization to the Hamiltonian Path problem is trivial. 
Hence we can derive the following result.

\begin{corollary}
    \label{theo:planar}
    The \feasn for planar networks with a max degree of 3 is strongly NP-hard.
\end{corollary}

\bibliographystyle{alpha}
\bibliography{braess_paradox}
\end{document}